\documentclass{article}
\usepackage{chl-alg}
\usepackage{vaucanson-g}
\usepackage{pstricks,pst-node}

\newtheorem{theorem}{Theorem}
\newtheorem{proposition}[theorem]{Proposition}

\newenvironment{proof}{\noindent\textbf{Proof.} }{\hfill\rule{2mm}{2mm}\medskip}

\def\dd{\mathinner{\ldotp\ldotp}}   % dot dot
\newcommand{\sle}{<\hspace{-4pt}<}  % strongly less (not prefix)
\def\sa#1{\mbox{\tt #1}}
\def\Lyn{\mathrm{Lyn}} % Longest-Lyndon table
\def\Rank{\mathrm{Rank}} % Ranks of suffixes, inverse of SA
\def\LCE{\mathrm{LCE}} % Longest Common Extension
\def\LeftChild{\texttt{leftchild}}
\def\RightChild{\texttt{rightchild}}
\def\Parent{\texttt{Parent}}
\def\NNS{\mathrm{NNS}}
\renewcommand{\fbox}[1]{#1}

\title{Cartesian trees and Lyndon trees}

\author{%
Maxime Crochemore%
\thanks{\texttt{Maxime.Crochemore@kcl.ac.uk}, King's College London, UK
 and Universit\'e Paris-Est, France.}
\and
Lu{\'\i}s M. S. Russo%
\thanks{\texttt{luis.russo@tecnico.ulisboa.pt},
  INESC-ID and the Department of Computer Science and
  Engineering,  Instituto Superior T\'{e}cnico, Universidade de
    Lisboa.
}
}

\begin{document}
\maketitle

%---------%---------%---------%---------%---------%---------%---------%--------%
\begin{abstract}
The article describes the structural and algorithmic relations between
Cartesian trees and Lyndon Trees. This leads to a uniform presentation of
the Lyndon table of a word corresponding to the Next Nearest Smaller table
of a sequence of numbers. It shows how to efficiently compute runs, that
is, maximal periodicities occurring in a word.
\end{abstract}

%---------%---------%---------%---------%---------%---------%---------%--------%
\section{Cartesian and Lyndon trees}\label{sect:intro}

The Cartesian tree, introduced by Vuillemin \cite{Vuillemin80} is a binary tree
associated with a sequence of numbers that label its nodes. It is both a
heap, with the smallest element at the root, and the sequence is recovered
during a symmetric traversal of the tree.

Cartesian tree have a series of applications in addition to that introduced
by Vuillemin \cite{Vuillemin80} on two-dimensional images. To quote a few of
them, they are used for range searching to implement range minimum queries in
a sequence of numbers through the help of Lowest Common Ancestor queries in
the Cartesian tree of the sequence \cite{GabowBT84}. They are also part of
sorting methods that want to take advantage of partially sorted subsequences
(see for example \cite{LevcopoulosP89}).

Lyndon trees are associated with Lyndon words, words that are
lexicographically smaller than all their proper non-empty suffixes (see
\cite{Lothaire83} and \cite{BerstelLRS08}). They also have several
interesting algorithmic applications and attracted much interest in
connection with the detection of runs (maximal periodicities) in words. The
notion of Lyndon roots of runs, introduced for cubic runs in
\cite{CrochemoreIKRRW12}, has led to the property that there is linear number
of square runs in a word. Originally conjectured by Kolpakov and
Kucherov~\cite{KolpakovK99}, it has eventually been proved by Bannai et
al.~\cite{BannaiIINTT15}. They also show how to compute efficiently all the
runs using implicitly the notion of Lyndon table (array), which is a side
product of the Lyndon tree construction.

%The structure of a Lyndon tree has been shown to be the same as the Cartesian
%tree of ranks of suffixes by Hohlweg and Reutenauer \cite{HohlwegR03}.
This article may be viewed as a follow-up of the publication by Hohlweg and
Reutenauer~\cite{HohlwegR03} in which they show the link between the two
types of trees. The bridge between them is a key property (stated in
Proposition~\ref{prop-main}) that relates a local condition on the factors
of the word to a global condition on its suffixes. It implies the structure
of a Lyndon tree is the same as the Cartesian tree of ranks of the
associated word suffixes.

%---------%---------%---------%---------%---------%---------%---------%--------%
\section{Cartesian tree}\label{sect:ct}

Let $x=(x[0], x[1], \dots, x[n-1])$ be a sequence of numbers of length $n$.
Below is a standard algorithm for computing its associated Cartesian tree.
Nodes of the tree are identified with positions of numbers on the sequence
and are labelled by the numbers with $X$. To simplify the algorithm we
insert a sentinel into the original sequence, i.e., we add a number $x[n] =
-\infty$. The purpose of this number is that it is smaller than any other
number that already exists in the sequence.

The algorithm proceeds from right to left, instead of left to right as usual,
to fit with the Lyndon tree construction. One step $i$ is to go up the
leftmost path of the tree from $i+1$ to find where to insert the node $i$.
(The artificial node $n$ acts as a sentinel to simplify the design.) During
the traversal, going to the parent of node $S$ is like going to the next
nearest value smaller than $x[i]$.

\medskip
\begin{algo}{CartesianTree}{x \textrm{ non-empty sequence of numbers of length } n}
    \SET{x[n]}{-\infty}
    \SET{X[n]}{x[n]}
    \SET{n.\texttt{LeftChild}}{\texttt{Null}}
    \DOFORD{i}{n-1}{0}
        \SET{S}{i+1}
        \DOWHILE{x[i]<X[S]} \label{algo-CT:5}
            \SET{S}{S.\Parent}
        \OD
        \SET{i.\RightChild}{S.\LeftChild}
        \SET{S.\LeftChild}{i}
    \OD
    \RETURN{\mbox{labelled built tree}}
\end{algo}

\medskip
The number of comparisons executed at line~\ref{algo-CT:5} is linear in $n$.
 Any comparison that yields $x[i] \geq X[S]$ means that the
\texttt{while} fails and therefore occurs at most once for each
$i$. Moreover the comparisons that yield $x[i] < X[S]$ for some position
$j$, i.e., $x[j] = X[S]$ implies that position $j$ will no longer be
involved in a latter comparison. An alternative view of this process is a
that the consequent $S \leftarrow S.\Parent$ assignment moves upward on the
rightmost branch of the current tree. Thus, the running time is $O(n)$.

The picture displays the Cartesion tree of the sequence of numbers:\\
$(7, 15, 12, 4, 10, 1, 5, 13, 6, 14, 11, 3, 9, 0, 2, 8, -\infty)$.

% Styles for the tree
\newpsstyle{dashed}{
  linestyle=dashed,
  dash=3pt,
  nodesepA=5pt
}
% Parent pointer Style
\newpsstyle{cppointer}{
  arrowsize=8pt,
  linewidth=2pt,
  border=4pt
}

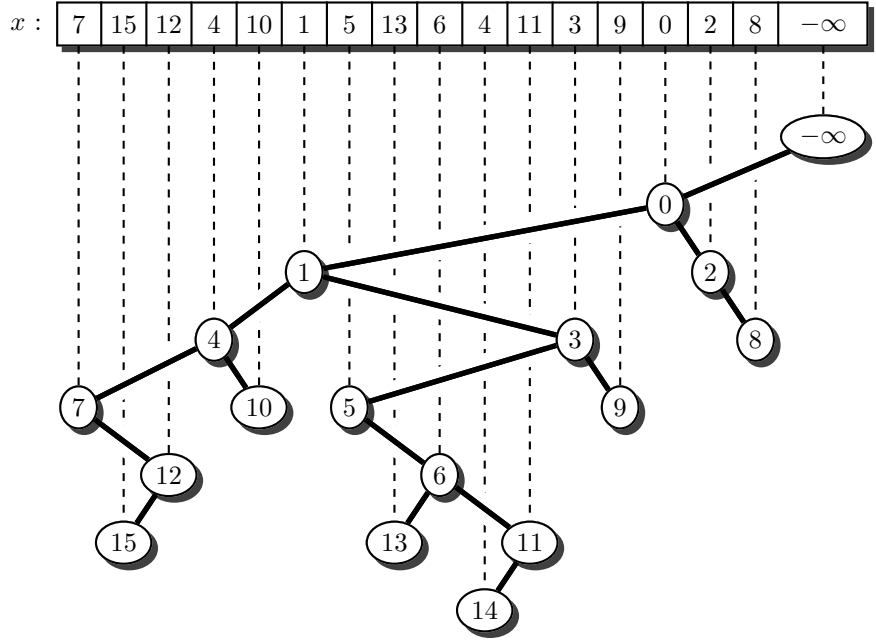
\newpsstyle{ppointer}{
  arrowsize=8pt,
  linewidth=2pt
}

% Spacing for the tree
\def\rx{0.6 }
\def\ry{0.9 }

\SpecialCoor

% The Cartesian tree
\begin{pspicture}[showgrid=false](-1,-8.5)(11,1)
% The array above the tree
% The lines
  \psframe[shadow=true](! -.5 \rx mul -.3)(! 17.5 \rx mul .3)
  \psline(! 0 0.5 add \rx mul -.3)(! 0 0.5 add \rx mul .3)
  \psline(! 1 0.5 add \rx mul -.3)(! 1 0.5 add \rx mul .3)
  \psline(! 2 0.5 add \rx mul -.3)(! 2 0.5 add \rx mul .3)
  \psline(! 3 0.5 add \rx mul -.3)(! 3 0.5 add \rx mul .3)
  \psline(! 4 0.5 add \rx mul -.3)(! 4 0.5 add \rx mul .3)
  \psline(! 5 0.5 add \rx mul -.3)(! 5 0.5 add \rx mul .3)
  \psline(! 6 0.5 add \rx mul -.3)(! 6 0.5 add \rx mul .3)
  \psline(! 7 0.5 add \rx mul -.3)(! 7 0.5 add \rx mul .3)
  \psline(! 8 0.5 add \rx mul -.3)(! 8 0.5 add \rx mul .3)
  \psline(! 9 0.5 add \rx mul -.3)(! 9 0.5 add \rx mul .3)
  \psline(! 10 0.5 add \rx mul -.3)(! 10 0.5 add \rx mul .3)
  \psline(! 11 0.5 add \rx mul -.3)(! 11 0.5 add \rx mul .3)
  \psline(! 12 0.5 add \rx mul -.3)(! 12 0.5 add \rx mul .3)
  \psline(! 13 0.5 add \rx mul -.3)(! 13 0.5 add \rx mul .3)
  \psline(! 14 0.5 add \rx mul -.3)(! 14 0.5 add \rx mul .3)
  \psline(! 15 0.5 add \rx mul -.3)(! 15 0.5 add \rx mul .3)
% The values
  \rput(-0.7,0){$x$\ :}

  \rput(! 0 \rx mul 0){\rnode{N0}{7}}
  \rput(! 1 \rx mul 0){\rnode{N1}{15}}
  \rput(! 2 \rx mul 0){\rnode{N2}{12}}
  \rput(! 3 \rx mul 0){\rnode{N3}{4}}
  \rput(! 4 \rx mul 0){\rnode{N4}{10}}
  \rput(! 5 \rx mul 0){\rnode{N5}{1}}
  \rput(! 6 \rx mul 0){\rnode{N6}{5}}
  \rput(! 7 \rx mul 0){\rnode{N7}{13}}
  \rput(! 8 \rx mul 0){\rnode{N8}{6}}
  \rput(! 9 \rx mul 0){\rnode{N9}{4}}
  \rput(! 10 \rx mul 0){\rnode{N10}{11}}
  \rput(! 11 \rx mul 0){\rnode{N11}{3}}
  \rput(! 12 \rx mul 0){\rnode{N12}{9}}
  \rput(! 13 \rx mul 0){\rnode{N13}{0}}
  \rput(! 14 \rx mul 0){\rnode{N14}{2}}
  \rput(! 15 \rx mul 0){\rnode{N15}{8}}
  \rput(! 16.5 \rx mul 0){\rnode{N16}{$-\infty$}}

  % The tree itself
\rput(0,-1.5){
  \rput(! 0 \rx mul -4 \ry mul){\ovalnode{T0}{7}}
  \rput(! 1 \rx mul -6 \ry mul){\ovalnode{T1}{15}}
  \rput(! 2 \rx mul -5 \ry mul){\ovalnode{T2}{12}}
  \rput(! 3 \rx mul -3 \ry mul){\ovalnode{T3}{4}}
  \rput(! 4 \rx mul -4 \ry mul){\ovalnode{T4}{10}}
  \rput(! 5 \rx mul -2 \ry mul){\ovalnode{T5}{1}}
  \rput(! 6 \rx mul -4 \ry mul){\ovalnode{T6}{5}}
  \rput(! 7 \rx mul -6 \ry mul){\ovalnode{T7}{13}}
  \rput(! 8 \rx mul -5 \ry mul){\ovalnode{T8}{6}}
  \rput(! 9 \rx mul -7 \ry mul){\ovalnode{T9}{14}}
  \rput(! 10 \rx mul -6 \ry mul){\ovalnode{T10}{11}}
  \rput(! 11 \rx mul -3 \ry mul){\ovalnode{T11}{3}}
  \rput(! 12 \rx mul -4 \ry mul){\ovalnode{T12}{9}}
  \rput(! 13 \rx mul -1 \ry mul){\ovalnode{T13}{0}}
  \rput(! 14 \rx mul -2 \ry mul){\ovalnode{T14}{2}}
  \rput(! 15 \rx mul -3 \ry mul){\ovalnode{T15}{8}}
  \rput(! 16.5 \rx mul -0 \ry mul){\ovalnode{T16}{$-\infty$}}
  }

% Link Array above with tree bellow
  \ncline[style=dashed]{N0}{T0}
  \ncline[style=dashed]{N1}{T1}
  \ncline[style=dashed]{N2}{T2}
  \ncline[style=dashed]{N3}{T3}
  \ncline[style=dashed]{N4}{T4}
  \ncline[style=dashed]{N5}{T5}
  \ncline[style=dashed]{N6}{T6}
  \ncline[style=dashed]{N7}{T7}
  \ncline[style=dashed]{N8}{T8}
  \ncline[style=dashed]{N9}{T9}
  \ncline[style=dashed]{N10}{T10}
  \ncline[style=dashed]{N11}{T11}
  \ncline[style=dashed]{N12}{T12}
  \ncline[style=dashed]{N13}{T13}
  \ncline[style=dashed]{N14}{T14}
  \ncline[style=dashed]{N15}{T15}
  \ncline[style=dashed]{N16}{T16}

  % Link the nodes of the tree

  \ncline[style=cppointer]{-}{T0}{T3}
  \ncline[style=cppointer]{-}{T1}{T2}
  \ncline[style=cppointer]{-}{T2}{T0}
  \ncline[style=cppointer]{-}{T3}{T5}
  \ncline[style=cppointer]{-}{T4}{T3}
  \ncline[style=cppointer]{-}{T5}{T13}
  \ncline[style=cppointer]{-}{T6}{T11}
  \ncline[style=cppointer]{-}{T7}{T8}
  \ncline[style=cppointer]{-}{T8}{T6}
  \ncline[style=cppointer]{-}{T9}{T10}
  \ncline[style=cppointer]{-}{T10}{T8}
  \ncline[style=cppointer]{-}{T11}{T5}
  \ncline[style=cppointer]{-}{T12}{T11}
  \ncline[style=cppointer]{-}{T13}{T16}
  \ncline[style=cppointer]{-}{T14}{T13}
  \ncline[style=cppointer]{-}{T15}{T14}
  %
  % Draw nodes again
  %
    \psset{shadow=true}
\rput(0,-1.5){
  \rput(! 0 \rx mul -4 \ry mul){\ovalnode{T0}{7}}
  \rput(! 1 \rx mul -6 \ry mul){\ovalnode{T1}{15}}
  \rput(! 2 \rx mul -5 \ry mul){\ovalnode{T2}{12}}
  \rput(! 3 \rx mul -3 \ry mul){\ovalnode{T3}{4}}
  \rput(! 4 \rx mul -4 \ry mul){\ovalnode{T4}{10}}
  \rput(! 5 \rx mul -2 \ry mul){\ovalnode{T5}{1}}
  \rput(! 6 \rx mul -4 \ry mul){\ovalnode{T6}{5}}
  \rput(! 7 \rx mul -6 \ry mul){\ovalnode{T7}{13}}
  \rput(! 8 \rx mul -5 \ry mul){\ovalnode{T8}{6}}
  \rput(! 9 \rx mul -7 \ry mul){\ovalnode{T9}{14}}
  \rput(! 10 \rx mul -6 \ry mul){\ovalnode{T10}{11}}
  \rput(! 11 \rx mul -3 \ry mul){\ovalnode{T11}{3}}
  \rput(! 12 \rx mul -4 \ry mul){\ovalnode{T12}{9}}
  \rput(! 13 \rx mul -1 \ry mul){\ovalnode{T13}{0}}
  \rput(! 14 \rx mul -2 \ry mul){\ovalnode{T14}{2}}
  \rput(! 15 \rx mul -3 \ry mul){\ovalnode{T15}{8}}
  \rput(! 16.5 \rx mul -0 \ry mul){\ovalnode{T16}{$-\infty$}}
  }
  \psset{shadow=false}

  % Draw edges again
  %
  \ncline[style=ppointer]{-}{T0}{T3}
  \ncline[style=ppointer]{-}{T1}{T2}
  \ncline[style=ppointer]{-}{T2}{T0}
  \ncline[style=ppointer]{-}{T3}{T5}
  \ncline[style=ppointer]{-}{T4}{T3}
  \ncline[style=ppointer]{-}{T5}{T13}
  \ncline[style=ppointer]{-}{T6}{T11}
  \ncline[style=ppointer]{-}{T7}{T8}
  \ncline[style=ppointer]{-}{T8}{T6}
  \ncline[style=ppointer]{-}{T9}{T10}
  \ncline[style=ppointer]{-}{T10}{T8}
  \ncline[style=ppointer]{-}{T11}{T5}
  \ncline[style=ppointer]{-}{T12}{T11}
  \ncline[style=ppointer]{-}{T13}{T16}
  \ncline[style=ppointer]{-}{T14}{T13}
  \ncline[style=ppointer]{-}{T15}{T14}

\end{pspicture}

The next picture exemplifies the \textsc{CartesianTree} Algorithm by
inserting the number $5$ into an Cartesian tree. At each step the algorithm
considers the leftmost branch of the tree, highlighted by the arrows in the
picture. These arrows represent the \texttt{parent} pointers that are used
by the inner cycle while $x[i] < X[S]$. In this example the while guard is
true twice when $5 < 13$ and $5 < 6$. The final comparison yields $5 > 3$
and therefore the cycle stops, notice that the \texttt{parent} pointers of
$13$ and $13$ are not represented by arrows in the second tree, as they are
no longer part of the leftmost branch.

% The Cartesian adding a number
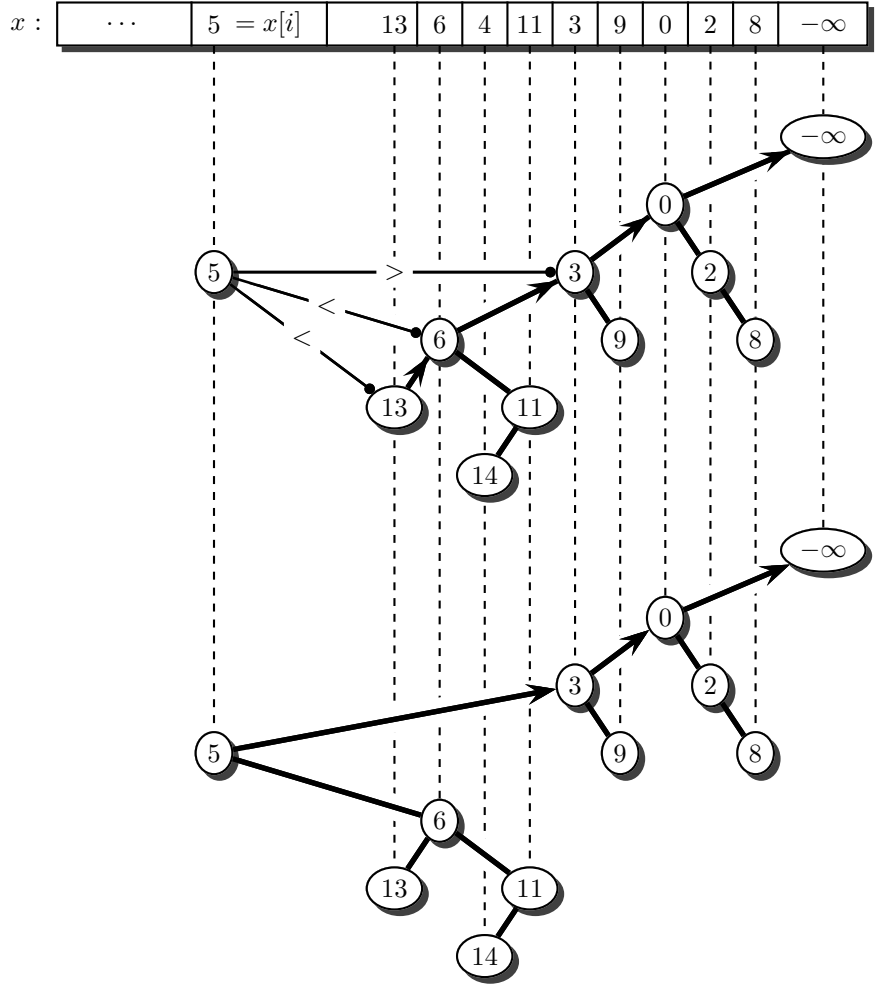
\begin{pspicture}[showgrid=false](-1,-13.0)(11,1)
% The array above the tree
% The lines
  \psframe[shadow=true](! -.5 \rx mul -.3)(! 17.5 \rx mul .3)
  \psline(! 2 0.5 add \rx mul -.3)(! 2 0.5 add \rx mul .3)
  \psline(! 5 0.5 add \rx mul -.3)(! 5 0.5 add \rx mul .3)
  \psline(! 7 0.5 add \rx mul -.3)(! 7 0.5 add \rx mul .3)
  \psline(! 8 0.5 add \rx mul -.3)(! 8 0.5 add \rx mul .3)
  \psline(! 9 0.5 add \rx mul -.3)(! 9 0.5 add \rx mul .3)
  \psline(! 10 0.5 add \rx mul -.3)(! 10 0.5 add \rx mul .3)
  \psline(! 11 0.5 add \rx mul -.3)(! 11 0.5 add \rx mul .3)
  \psline(! 12 0.5 add \rx mul -.3)(! 12 0.5 add \rx mul .3)
  \psline(! 13 0.5 add \rx mul -.3)(! 13 0.5 add \rx mul .3)
  \psline(! 14 0.5 add \rx mul -.3)(! 14 0.5 add \rx mul .3)
  \psline(! 15 0.5 add \rx mul -.3)(! 15 0.5 add \rx mul .3)
% The values
  \rput(! 4.2 \rx mul 0){$= x[i]$}
  \rput(-0.7,0){$x$\ :}
  \rput(! 1 \rx mul 0){\ldots}
  \rput(! 3 \rx mul 0){\rnode{N6}{5}}
  \rput(! 7 \rx mul 0){\rnode{N7}{13}}
  \rput(! 8 \rx mul 0){\rnode{N8}{6}}
  \rput(! 9 \rx mul 0){\rnode{N9}{4}}
  \rput(! 10 \rx mul 0){\rnode{N10}{11}}
  \rput(! 11 \rx mul 0){\rnode{N11}{3}}
  \rput(! 12 \rx mul 0){\rnode{N12}{9}}
  \rput(! 13 \rx mul 0){\rnode{N13}{0}}
  \rput(! 14 \rx mul 0){\rnode{N14}{2}}
  \rput(! 15 \rx mul 0){\rnode{N15}{8}}
  \rput(! 16.5 \rx mul 0){\rnode{N16}{$-\infty$}}

  %% The second tree

  % The tree itself
\rput(0,-7.0){
  \rput(! 3 \rx mul -3 \ry mul){\ovalnode{T6}{5}}
  \rput(! 7 \rx mul -5 \ry mul){\ovalnode{T7}{13}}
  \rput(! 8 \rx mul -4 \ry mul){\ovalnode{T8}{6}}
  \rput(! 9 \rx mul -6 \ry mul){\ovalnode{T9}{14}}
  \rput(! 10 \rx mul -5 \ry mul){\ovalnode{T10}{11}}
  \rput(! 11 \rx mul -2 \ry mul){\ovalnode{T11}{3}}
  \rput(! 12 \rx mul -3 \ry mul){\ovalnode{T12}{9}}
  \rput(! 13 \rx mul -1 \ry mul){\ovalnode{T13}{0}}
  \rput(! 14 \rx mul -2 \ry mul){\ovalnode{T14}{2}}
  \rput(! 15 \rx mul -3 \ry mul){\ovalnode{T15}{8}}
  \rput(! 16.5 \rx mul -0 \ry mul){\ovalnode{T16}{$-\infty$}}
  }

  %
  % Link Array above with tree bellow
  \ncline[style=dashed]{N6}{T6}
  \ncline[style=dashed]{N7}{T7}
  \ncline[style=dashed]{N8}{T8}
  \ncline[style=dashed]{N9}{T9}
  \ncline[style=dashed]{N10}{T10}
  \ncline[style=dashed]{N11}{T11}
  \ncline[style=dashed]{N12}{T12}
  \ncline[style=dashed]{N13}{T13}
  \ncline[style=dashed]{N14}{T14}
  \ncline[style=dashed]{N15}{T15}
  \ncline[style=dashed]{N16}{T16}

  % Link the nodes of the tree

  \ncline[style=cppointer]{->}{T6}{T11}
  \ncline[style=cppointer]{-}{T7}{T8}
  \ncline[style=cppointer]{-}{T8}{T6}
  \ncline[style=cppointer]{-}{T9}{T10}
  \ncline[style=cppointer]{-}{T10}{T8}
  \ncline[style=cppointer]{->}{T11}{T13}
  \ncline[style=cppointer]{-}{T12}{T11}
  \ncline[style=cppointer]{->}{T13}{T16}
  \ncline[style=cppointer]{-}{T14}{T13}
  \ncline[style=cppointer]{-}{T15}{T14}
  %
  % Draw nodes again
  %
    \psset{shadow=true}
\rput(0,-7.0){
  \rput(! 3 \rx mul -3 \ry mul){\ovalnode{T6}{5}}
  \rput(! 7 \rx mul -5 \ry mul){\ovalnode{T7}{13}}
  \rput(! 8 \rx mul -4 \ry mul){\ovalnode{T8}{6}}
  \rput(! 9 \rx mul -6 \ry mul){\ovalnode{T9}{14}}
  \rput(! 10 \rx mul -5 \ry mul){\ovalnode{T10}{11}}
  \rput(! 11 \rx mul -2 \ry mul){\ovalnode{T11}{3}}
  \rput(! 12 \rx mul -3 \ry mul){\ovalnode{T12}{9}}
  \rput(! 13 \rx mul -1 \ry mul){\ovalnode{T13}{0}}
  \rput(! 14 \rx mul -2 \ry mul){\ovalnode{T14}{2}}
  \rput(! 15 \rx mul -3 \ry mul){\ovalnode{T15}{8}}
  \rput(! 16.5 \rx mul -0 \ry mul){\ovalnode{T16}{$-\infty$}}
  }
  \psset{shadow=false}

  % Draw edges again
  \ncline[style=ppointer]{->}{T6}{T11}
  \ncline[style=ppointer]{-}{T7}{T8}
  \ncline[style=ppointer]{-}{T8}{T6}
  \ncline[style=ppointer]{-}{T9}{T10}
  \ncline[style=ppointer]{-}{T10}{T8}
  \ncline[style=ppointer]{->}{T11}{T13}
  \ncline[style=ppointer]{-}{T12}{T11}
  \ncline[style=ppointer]{->}{T13}{T16}
  \ncline[style=ppointer]{-}{T14}{T13}
  \ncline[style=ppointer]{-}{T15}{T14}

  %
  % The tree above

    % The tree itself
\rput(0,-1.5){
  \rput(! 3 \rx mul -2 \ry mul){\ovalnode{T6}{5}}
  \rput(! 7 \rx mul -4 \ry mul){\ovalnode{T7}{13}}
  \rput(! 8 \rx mul -3 \ry mul){\ovalnode{T8}{6}}
  \rput(! 9 \rx mul -5 \ry mul){\ovalnode{T9}{14}}
  \rput(! 10 \rx mul -4 \ry mul){\ovalnode{T10}{11}}
  \rput(! 11 \rx mul -2 \ry mul){\ovalnode{T11}{3}}
  \rput(! 12 \rx mul -3 \ry mul){\ovalnode{T12}{9}}
  \rput(! 13 \rx mul -1 \ry mul){\ovalnode{T13}{0}}
  \rput(! 14 \rx mul -2 \ry mul){\ovalnode{T14}{2}}
  \rput(! 15 \rx mul -3 \ry mul){\ovalnode{T15}{8}}
  \rput(! 16.5 \rx mul -0 \ry mul){\ovalnode{T16}{$-\infty$}}
  }

  % Link the nodes of the tree

%  \ncline[style=cppointer]{->}{T6}{T11}
  \ncline[linewidth=1pt,border=4pt]{-**}{T6}{T7}
  \ncline[linewidth=1pt,border=4pt]{-**}{T6}{T8}
  \ncline[linewidth=1pt,border=4pt]{-**}{T6}{T11}
  \ncline[style=cppointer]{->}{T7}{T8}
  \ncline[style=cppointer]{->}{T8}{T11}
  \ncline[style=cppointer]{-}{T9}{T10}
  \ncline[style=cppointer]{-}{T10}{T8}
  \ncline[style=cppointer]{->}{T11}{T13}
  \ncline[style=cppointer]{-}{T12}{T11}
  \ncline[style=cppointer]{->}{T13}{T16}
  \ncline[style=cppointer]{-}{T14}{T13}
  \ncline[style=cppointer]{-}{T15}{T14}
  %
  % Draw nodes again
  %
    \psset{shadow=true}
\rput(0,-1.5){
  \rput(! 3 \rx mul -2 \ry mul){\ovalnode{T6}{5}}
  \rput(! 7 \rx mul -4 \ry mul){\ovalnode{T7}{13}}
  \rput(! 8 \rx mul -3 \ry mul){\ovalnode{T8}{6}}
  \rput(! 9 \rx mul -5 \ry mul){\ovalnode{T9}{14}}
  \rput(! 10 \rx mul -4 \ry mul){\ovalnode{T10}{11}}
  \rput(! 11 \rx mul -2 \ry mul){\ovalnode{T11}{3}}
  \rput(! 12 \rx mul -3 \ry mul){\ovalnode{T12}{9}}
  \rput(! 13 \rx mul -1 \ry mul){\ovalnode{T13}{0}}
  \rput(! 14 \rx mul -2 \ry mul){\ovalnode{T14}{2}}
  \rput(! 15 \rx mul -3 \ry mul){\ovalnode{T15}{8}}
  \rput(! 16.5 \rx mul -0 \ry mul){\ovalnode{T16}{$-\infty$}}
  }
  \psset{shadow=false}

  % Draw edges again
  \ncline[linewidth=1pt]{-}{T6}{T7} \ncput*{$<$}
  \ncline[linewidth=1pt]{-}{T6}{T8} \ncput*{$<$}
  \ncline[linewidth=1pt]{-}{T6}{T11} \ncput*{$>$}
  \ncline[style=ppointer]{-}{T7}{T8}
  \ncline[style=ppointer]{-}{T8}{T11}
  \ncline[style=ppointer]{-}{T9}{T10}
  \ncline[style=ppointer]{-}{T10}{T8}
  \ncline[style=ppointer]{-}{T11}{T13}
  \ncline[style=ppointer]{-}{T12}{T11}
  \ncline[style=ppointer]{-}{T13}{T16}
  \ncline[style=ppointer]{-}{T14}{T13}
  \ncline[style=ppointer]{-}{T15}{T14}

\end{pspicture}

This moving upwards process computed by the inner cycle, corresponds to
finding the Next nearest smaller value. In our example, when given the
number $5$ we searched the sequence until we reached the number $3$, note
that by moving upwards on the tree this process is faster than computing
linear scan from right to left, in particular we did not compare with the
numbers $14$ and $11$.

%---------%---------%---------%---------%---------%---------%---------%--------%
\paragraph{Next nearest smaller table}

The Next nearest smaller table $\NNS$ of a (non-empty) sequence $y$ of
numbers is defined as follows. For a position $i$ on $x$, $i=0,\dots,|x|-1$,
$\NNS[i]$ is the smallest position $j>i$ of an element $x[j]<x[i]$, or $n$ if
none exist:
\[\NNS[i] = \min\{j \mid x[j]<x[i]\}\cup\{n\}.\]

The following picture shows the $\NNS$ table illustrated over the Cartesian
tree. We show the corresponding table bellow the tree. Moreover each node
also shows an arrow to point to the corresponding Next nearest smaller
node. When a node is a left child of its parent the arrows are simply the
\texttt{parent} pointers. However when the node is a right child of its
parent then the arrows are shown with dashed lines and point to an ancestor
of the node that is to its right.

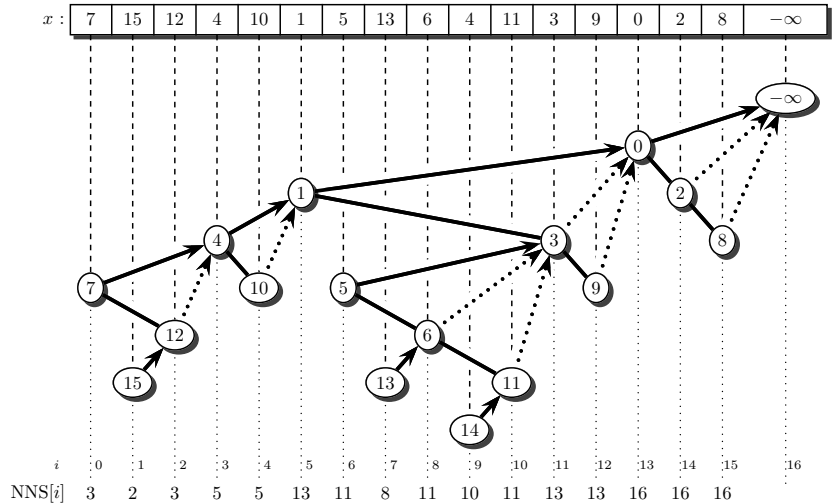
\def\rx{0.8 }
% The Cartesian tree with NNS

\begin{Scalebox}{0.7}
\begin{pspicture}[showgrid=false](-1.5,-9.5)(14.5,0.5)
% The array above the tree
% The lines
  \psframe[shadow=true](! -.5 \rx mul -.3)(! 17.5 \rx mul .3)
  \psline(! 0 0.5 add \rx mul -.3)(! 0 0.5 add \rx mul .3)
  \psline(! 1 0.5 add \rx mul -.3)(! 1 0.5 add \rx mul .3)
  \psline(! 2 0.5 add \rx mul -.3)(! 2 0.5 add \rx mul .3)
  \psline(! 3 0.5 add \rx mul -.3)(! 3 0.5 add \rx mul .3)
  \psline(! 4 0.5 add \rx mul -.3)(! 4 0.5 add \rx mul .3)
  \psline(! 5 0.5 add \rx mul -.3)(! 5 0.5 add \rx mul .3)
  \psline(! 6 0.5 add \rx mul -.3)(! 6 0.5 add \rx mul .3)
  \psline(! 7 0.5 add \rx mul -.3)(! 7 0.5 add \rx mul .3)
  \psline(! 8 0.5 add \rx mul -.3)(! 8 0.5 add \rx mul .3)
  \psline(! 9 0.5 add \rx mul -.3)(! 9 0.5 add \rx mul .3)
  \psline(! 10 0.5 add \rx mul -.3)(! 10 0.5 add \rx mul .3)
  \psline(! 11 0.5 add \rx mul -.3)(! 11 0.5 add \rx mul .3)
  \psline(! 12 0.5 add \rx mul -.3)(! 12 0.5 add \rx mul .3)
  \psline(! 13 0.5 add \rx mul -.3)(! 13 0.5 add \rx mul .3)
  \psline(! 14 0.5 add \rx mul -.3)(! 14 0.5 add \rx mul .3)
  \psline(! 15 0.5 add \rx mul -.3)(! 15 0.5 add \rx mul .3)
% The values
  \rput(-0.7,0){$x$\ :}

  \rput(! 0 \rx mul 0){\rnode{N0}{7}}
  \rput(! 1 \rx mul 0){\rnode{N1}{15}}
  \rput(! 2 \rx mul 0){\rnode{N2}{12}}
  \rput(! 3 \rx mul 0){\rnode{N3}{4}}
  \rput(! 4 \rx mul 0){\rnode{N4}{10}}
  \rput(! 5 \rx mul 0){\rnode{N5}{1}}
  \rput(! 6 \rx mul 0){\rnode{N6}{5}}
  \rput(! 7 \rx mul 0){\rnode{N7}{13}}
  \rput(! 8 \rx mul 0){\rnode{N8}{6}}
  \rput(! 9 \rx mul 0){\rnode{N9}{4}}
  \rput(! 10 \rx mul 0){\rnode{N10}{11}}
  \rput(! 11 \rx mul 0){\rnode{N11}{3}}
  \rput(! 12 \rx mul 0){\rnode{N12}{9}}
  \rput(! 13 \rx mul 0){\rnode{N13}{0}}
  \rput(! 14 \rx mul 0){\rnode{N14}{2}}
  \rput(! 15 \rx mul 0){\rnode{N15}{8}}
  \rput(! 16.5 \rx mul 0){\rnode{N16}{$-\infty$}}

\def\NNS{\mathrm{NNS}}

  % The NNS table
  \rput(! -1.0 -10 \ry mul){$\NNS[i]$}
  \rput(! 0 \rx mul -10 \ry mul){\rnode{L0}{3}}
  \rput(! 1 \rx mul -10 \ry mul){\rnode{L1}{2}}
  \rput(! 2 \rx mul -10 \ry mul){\rnode{L2}{3}}
  \rput(! 3 \rx mul -10 \ry mul){\rnode{L3}{5}}
  \rput(! 4 \rx mul -10 \ry mul){\rnode{L4}{5}}
  \rput(! 5 \rx mul -10 \ry mul){\rnode{L5}{13}}
  \rput(! 6 \rx mul -10 \ry mul){\rnode{L6}{11}}
  \rput(! 7 \rx mul -10 \ry mul){\rnode{L7}{8}}
  \rput(! 8 \rx mul -10 \ry mul){\rnode{L8}{11}}
  \rput(! 9 \rx mul -10 \ry mul){\rnode{L9}{10}}
  \rput(! 10 \rx mul -10 \ry mul){\rnode{L10}{11}}
  \rput(! 11 \rx mul -10 \ry mul){\rnode{L11}{13}}
  \rput(! 12 \rx mul -10 \ry mul){\rnode{L12}{13}}
  \rput(! 13 \rx mul -10 \ry mul){\rnode{L13}{16}}
  \rput(! 14 \rx mul -10 \ry mul){\rnode{L14}{16}}
  \rput(! 15 \rx mul -10 \ry mul){\rnode{L15}{16}}
  \rput(! 16.5 \rx mul -10 \ry mul){\rnode{L16}{}}

  % The tree itself
\rput(0,-1.5){
  \rput(! 0 \rx mul -4 \ry mul){\ovalnode{T0}{7}}
  \rput(! 1 \rx mul -6 \ry mul){\ovalnode{T1}{15}}
  \rput(! 2 \rx mul -5 \ry mul){\ovalnode{T2}{12}}
  \rput(! 3 \rx mul -3 \ry mul){\ovalnode{T3}{4}}
  \rput(! 4 \rx mul -4 \ry mul){\ovalnode{T4}{10}}
  \rput(! 5 \rx mul -2 \ry mul){\ovalnode{T5}{1}}
  \rput(! 6 \rx mul -4 \ry mul){\ovalnode{T6}{5}}
  \rput(! 7 \rx mul -6 \ry mul){\ovalnode{T7}{13}}
  \rput(! 8 \rx mul -5 \ry mul){\ovalnode{T8}{6}}
  \rput(! 9 \rx mul -7 \ry mul){\ovalnode{T9}{14}}
  \rput(! 10 \rx mul -6 \ry mul){\ovalnode{T10}{11}}
  \rput(! 11 \rx mul -3 \ry mul){\ovalnode{T11}{3}}
  \rput(! 12 \rx mul -4 \ry mul){\ovalnode{T12}{9}}
  \rput(! 13 \rx mul -1 \ry mul){\ovalnode{T13}{0}}
  \rput(! 14 \rx mul -2 \ry mul){\ovalnode{T14}{2}}
  \rput(! 15 \rx mul -3 \ry mul){\ovalnode{T15}{8}}
  \rput(! 16.5 \rx mul -0 \ry mul){\ovalnode{T16}{$-\infty$}}
  }

% Link Array above with tree bellow
  \ncline[style=dashed]{N0}{T0}
  \ncline[style=dashed]{N1}{T1}
  \ncline[style=dashed]{N2}{T2}
  \ncline[style=dashed]{N3}{T3}
  \ncline[style=dashed]{N4}{T4}
  \ncline[style=dashed]{N5}{T5}
  \ncline[style=dashed]{N6}{T6}
  \ncline[style=dashed]{N7}{T7}
  \ncline[style=dashed]{N8}{T8}
  \ncline[style=dashed]{N9}{T9}
  \ncline[style=dashed]{N10}{T10}
  \ncline[style=dashed]{N11}{T11}
  \ncline[style=dashed]{N12}{T12}
  \ncline[style=dashed]{N13}{T13}
  \ncline[style=dashed]{N14}{T14}
  \ncline[style=dashed]{N15}{T15}
  \ncline[style=dashed]{N16}{T16}

% Link Array bellow with tree above
  \ncline[linestyle=dotted,nodesepA=5pt]{L0}{T0}
  \ncline[linestyle=dotted,nodesepA=5pt]{L1}{T1}
  \ncline[linestyle=dotted,nodesepA=5pt]{L2}{T2}
  \ncline[linestyle=dotted,nodesepA=5pt]{L3}{T3}
  \ncline[linestyle=dotted,nodesepA=5pt]{L4}{T4}
  \ncline[linestyle=dotted,nodesepA=5pt]{L5}{T5}
  \ncline[linestyle=dotted,nodesepA=5pt]{L6}{T6}
  \ncline[linestyle=dotted,nodesepA=5pt]{L7}{T7}
  \ncline[linestyle=dotted,nodesepA=5pt]{L8}{T8}
  \ncline[linestyle=dotted,nodesepA=5pt]{L9}{T9}
  \ncline[linestyle=dotted,nodesepA=5pt]{L10}{T10}
  \ncline[linestyle=dotted,nodesepA=5pt]{L11}{T11}
  \ncline[linestyle=dotted,nodesepA=5pt]{L12}{T12}
  \ncline[linestyle=dotted,nodesepA=5pt]{L13}{T13}
  \ncline[linestyle=dotted,nodesepA=5pt]{L14}{T14}
  \ncline[linestyle=dotted,nodesepA=5pt]{L15}{T15}
  \ncline[linestyle=dotted,nodesepA=5pt]{L16}{T16}

  % Link the nodes of the tree

  \ncline[style=cppointer]{->}{T0}{T3}
  \ncline[style=cppointer]{->}{T1}{T2}
  \ncline[style=cppointer]{-}{T2}{T0}
  \ncline[style=cppointer,linestyle=dotted]{->}{T2}{T3}
  \ncline[style=cppointer]{->}{T3}{T5}
  \ncline[style=cppointer]{-}{T4}{T3}
  \ncline[style=cppointer,linestyle=dotted]{->}{T4}{T5}
  \ncline[style=cppointer]{->}{T5}{T13}
  \ncline[style=cppointer]{->}{T6}{T11}
  \ncline[style=cppointer]{->}{T7}{T8}
  \ncline[style=cppointer]{-}{T8}{T6}
  \ncline[style=cppointer,linestyle=dotted]{->}{T8}{T11}
  \ncline[style=cppointer]{->}{T9}{T10}
  \ncline[style=cppointer]{-}{T10}{T8}
  \ncline[style=cppointer,linestyle=dotted]{->}{T10}{T11}
  \ncline[style=cppointer]{-}{T11}{T5}
  \ncline[style=cppointer,linestyle=dotted]{->}{T11}{T13}
  \ncline[style=cppointer]{-}{T12}{T11}
  \ncline[style=cppointer,linestyle=dotted]{->}{T12}{T13}
  \ncline[style=cppointer]{->}{T13}{T16}
  \ncline[style=cppointer]{-}{T14}{T13}
  \ncline[style=cppointer,linestyle=dotted]{->}{T14}{T16}
  \ncline[style=cppointer]{-}{T15}{T14}
  \ncline[style=cppointer,linestyle=dotted]{->}{T15}{T16}

  \rput(! -1.0 0.2 add \rx mul -9.4 \ry mul){\scriptsize $i$}

  \rput(! 0 0.2 add \rx mul -9.4 \ry mul){\scriptsize 0}
  \rput(!   1 0.2 add \rx mul -9.4 \ry mul){\scriptsize   1}
  \rput(!   2 0.2 add \rx mul -9.4 \ry mul){\scriptsize   2}
  \rput(!   3 0.2 add \rx mul -9.4 \ry mul){\scriptsize   3}
  \rput(!   4 0.2 add \rx mul -9.4 \ry mul){\scriptsize   4}
  \rput(!   5 0.2 add \rx mul -9.4 \ry mul){\scriptsize   5}
  \rput(!   6 0.2 add \rx mul -9.4 \ry mul){\scriptsize   6}
  \rput(!   7 0.2 add \rx mul -9.4 \ry mul){\scriptsize   7}
  \rput(!   8 0.2 add \rx mul -9.4 \ry mul){\scriptsize   8}
  \rput(!   9 0.2 add \rx mul -9.4 \ry mul){\scriptsize   9}
  \rput(!   10 0.2 add \rx mul -9.4 \ry mul){\scriptsize   10}
  \rput(!   11 0.2 add \rx mul -9.4 \ry mul){\scriptsize   11}
  \rput(!   12 0.2 add \rx mul -9.4 \ry mul){\scriptsize   12}
  \rput(!   13 0.2 add \rx mul -9.4 \ry mul){\scriptsize   13}
  \rput(!   14 0.2 add \rx mul -9.4 \ry mul){\scriptsize   14}
  \rput(!   15 0.2 add \rx mul -9.4 \ry mul){\scriptsize   15}
  \rput(!   16.5 0.2 add \rx mul -9.4 \ry mul){\scriptsize   16}

  % Draw nodes again
  %
    \psset{shadow=true}
\rput(0,-1.5){
  \rput(! 0 \rx mul -4 \ry mul){\ovalnode{T0}{7}}
  \rput(! 1 \rx mul -6 \ry mul){\ovalnode{T1}{15}}
  \rput(! 2 \rx mul -5 \ry mul){\ovalnode{T2}{12}}
  \rput(! 3 \rx mul -3 \ry mul){\ovalnode{T3}{4}}
  \rput(! 4 \rx mul -4 \ry mul){\ovalnode{T4}{10}}
  \rput(! 5 \rx mul -2 \ry mul){\ovalnode{T5}{1}}
  \rput(! 6 \rx mul -4 \ry mul){\ovalnode{T6}{5}}
  \rput(! 7 \rx mul -6 \ry mul){\ovalnode{T7}{13}}
  \rput(! 8 \rx mul -5 \ry mul){\ovalnode{T8}{6}}
  \rput(! 9 \rx mul -7 \ry mul){\ovalnode{T9}{14}}
  \rput(! 10 \rx mul -6 \ry mul){\ovalnode{T10}{11}}
  \rput(! 11 \rx mul -3 \ry mul){\ovalnode{T11}{3}}
  \rput(! 12 \rx mul -4 \ry mul){\ovalnode{T12}{9}}
  \rput(! 13 \rx mul -1 \ry mul){\ovalnode{T13}{0}}
  \rput(! 14 \rx mul -2 \ry mul){\ovalnode{T14}{2}}
  \rput(! 15 \rx mul -3 \ry mul){\ovalnode{T15}{8}}
  \rput(! 16.5 \rx mul -0 \ry mul){\ovalnode{T16}{$-\infty$}}
  }
  \psset{shadow=false}

  % Draw edges again
  %
  \ncline[style=ppointer]{->}{T0}{T3}
  \ncline[style=ppointer]{->}{T1}{T2}
  \ncline[style=ppointer]{-}{T2}{T0}
  \ncline[style=ppointer,linestyle=dotted]{->}{T2}{T3}
  \ncline[style=ppointer]{->}{T3}{T5}
  \ncline[style=ppointer]{-}{T4}{T3}
  \ncline[style=ppointer,linestyle=dotted]{->}{T4}{T5}
  \ncline[style=ppointer]{->}{T5}{T13}
  \ncline[style=ppointer]{->}{T6}{T11}
  \ncline[style=ppointer]{->}{T7}{T8}
  \ncline[style=ppointer]{-}{T8}{T6}
  \ncline[style=ppointer,linestyle=dotted]{->}{T8}{T11}
  \ncline[style=ppointer]{->}{T9}{T10}
  \ncline[style=ppointer]{-}{T10}{T8}
  \ncline[style=ppointer,linestyle=dotted]{->}{T10}{T11}
  \ncline[style=ppointer]{-}{T11}{T5}
  \ncline[style=ppointer,linestyle=dotted]{->}{T11}{T13}
  \ncline[style=ppointer]{-}{T12}{T11}
  \ncline[style=ppointer,linestyle=dotted]{->}{T12}{T13}
  \ncline[style=ppointer]{->}{T13}{T16}
  \ncline[style=ppointer]{-}{T14}{T13}
  \ncline[style=ppointer,linestyle=dotted]{->}{T14}{T16}
  \ncline[style=ppointer]{-}{T15}{T14}
  \ncline[style=ppointer,linestyle=dotted]{->}{T15}{T16}

\end{pspicture}
\end{Scalebox}

\medskip

It is interesting to notice that the algorithm used for constructing
Cartesian trees can be adapted to compute the $\NNS$ values. As illustrated
by the picture when a node is a left child then the $\NNS$ value is
actually a pointer to its parent on the tree. Now recall the
\textsc{CartesianTree} Algorithm and notice that whenever a value is
inserted in the tree it is always a left child and therefore its parent,
when it gets inserted, is the corresponding $\NNS$ value. Recall our
example when the value of $x[6] = 5$ is inserted into the tree it becomes
the left child of node $11$, with $x[11] = 3$, therefore $\NNS[6] =
11$. Note also in this example that when $x[6]$ is processed we have that
the node $8$ with $x[8] = 6$ was a left child of node $11$ with $x[11] = 3$
before the insertion but becomes a right child after the insertion. Still
the value $\NNS[8] = 11$ is not altered by this procedure.

The following modification of the \textsc{CartesianTree} Algorithm uses
this information to obtain the $\NNS$ values. Likewise it also runs in
linear time.

\medskip
\begin{algo}{NextNearestSmaller}{x \textrm{ non-empty sequence of numbers of length } n}
    \SET{(x[n],\NNS[n-1])}{(-\infty,n)}
    \DOFORD{i}{n-2}{0}
        \SET{j}{x[i+1]}
        \DOWHILE{x[i]<x[j]} \label{algo-NNS:5}
            \SET{j}{\NNS[j]}
        \OD
        \SET{\NNS[i]}{j}
    \OD
    \RETURN{\NNS}
\end{algo}

%---------%---------%---------%---------%---------%---------%---------%--------%
\section{Lyndon tree}\label{sect:lt}

Lyndon trees are associated with Lyndon words. Recall that a Lyndon word is a
non-empty word lexicographically smaller than all its proper non-empty
suffixes. The Lyndon tree of a Lyndon word $y$ corresponds recursively to the
following suffix (or standard) factorisation of $y$ when not reduced to a
single letter: $y$ can be written $uv$ where $v$ is chosen as the smallest
proper non-empty suffix of $y$. The word $u$ is then also a Lyndon word (see
\cite{Lothaire83}).

Algorithm \textsc{LyndonTree} builds the Lyndon tree of a Lyndon word $y$.
The hypothesis on $y$ is not a significant restriction because any word can
be turned into a Lyndon word by prepending to it a letter smaller than all
letters occurring in it. Otherwise, since any word factorises uniquely into
Lyndon words, the algorithm can produce the forest of Lyndon trees of the
factors.

The algorithm proceeds naturally from right to left on $y$ to find the
longest Lyndon word starting at each position $i$. It applies a known
property: if $u$ and $v$ are Lyndon words and $u<v$ then $uv$ is also a Lyndon
word with $u<uv<v$.

To facilitate the presentation, variable $u$ stores a phrase, that is, the
occurrence of a Lyndon factor of $y$ though the position of the factor is not
explicitly given, and $T(u)$ is the Lyndon tree associated with this
occurrence. Idem for $v$.

\medskip
\begin{algo}{LyndonTree}{y \textrm{ Lyndon word of length } n}
%  \SET{v}{y[n-1]}
  \SET{(v,T(v))}{(y[n-1],(y[n-1]))}
  \DOFORD{i}{n-2}{0}
%    \SET{u}{y[i]}
    \SET{(u,T(u))}{(y[i],(y[i]))}
    \DOWHILE{u < v}  \label{algo-LT:4}
      \SET{T(uv)}{(\mbox{new node},T(u),T(v))} \label{algo-LT:5}
      \SET{u}{uv} \label{lynt:7}
      \SET{v}{\textrm{next phrase, empty word if none}} \label{algo-LT:7}
    \OD
  \OD
  \RETURN{T(y)}
\end{algo}

\medskip
If the comparison $u<v$ at line~\ref{algo-LT:4} is done by mere letter
comparisons, the algorithm may run in quadratic time, for example if applied
on $y = \texttt{a}^k\texttt{b}\texttt{a}^k\texttt{c}$ (each factor
$\texttt{a}^i\texttt{b}$ is compared with the prefix $\texttt{a}^{i+1}$ of
$\texttt{a}^k\texttt{c}$ or with $\texttt{a}^k\texttt{c}$ itself).

However the algorithm can be implemented to run in linear time if the test
$u<v$ at line~\ref{algo-LT:4} is done in constant time because each execution
of instructions at lines~\ref{algo-LT:5}-\ref{algo-LT:7} decreases the number
of Lyndon phrases, which goes from $n$ to $1$.

%It is realised with the suffix array of $y$, RMQ, etc. In the comparison
%model, it is also realised if we know $\lcp(u,v)$ because the comparison
%resumes to a single letter comparison.

\def\rx{0.6 }
\begin{center}
\begin{pspicture}[showgrid=false](-1,-16)(11,1)

\psset{linewidth=1.5pt}
% level 0
  \psframe[framearc=0.4](! -.3 0 add \rx mul -.5)(! 0.3 16 add \rx mul 0.5)

% level 1
  \psframe[framearc=0.4](! -.3 0 add \rx mul -.5 -2 add)
  (! 0.3 13 add \rx mul 0.5 -2 add)
  \psframe[framearc=0.4](! -.3 14 add \rx mul -.5 -2 add)
  (! 0.3 16 add \rx mul 0.5 -2 add)

% level 2
  \psframe[framearc=0.4](! -.3 0 add \rx mul -.5 -4 add)
  (! 0.3 5 add \rx mul 0.5 -4 add)
  \psframe[framearc=0.4](! -.3 6 add \rx mul -.5 -4 add)
  (! 0.3 13 add \rx mul 0.5 -4 add)
  \psframe[framearc=0.4](! -.3 14 add \rx mul -.5 -4 add)
  (! 0.3 14 add \rx mul 0.5 -4 add)
  \psframe[framearc=0.4](! -.3 15 add \rx mul -.5 -4 add)
  (! 0.3 16 add \rx mul 0.5 -4 add)

% level 3
 \psframe[framearc=0.4](! -.3 0 add \rx mul -.5 -6 add)
 (! 0.3 3 add \rx mul 0.5 -6 add)
 \psframe[framearc=0.4](! -.3 4 add \rx mul -.5 -6 add)
 (! 0.3 5 add \rx mul 0.5 -6 add)
  \psframe[framearc=0.4](! -.3 6 add \rx mul -.5 -6 add)
  (! 0.3 11 add \rx mul 0.5 -6 add)
  \psframe[framearc=0.4](! -.3 12 add \rx mul -.5 -6 add)
  (! 0.3 13 add \rx mul 0.5 -6 add)

  \psframe[framearc=0.4](! -.3 15 add \rx mul -.5 -6 add)
  (! 0.3 15 add \rx mul 0.5 -6 add)
  \psframe[framearc=0.4](! -.3 16 add \rx mul -.5 -6 add)
  (! 0.3 16 add \rx mul 0.5 -6 add)

% level 4

\psframe[framearc=0.4](! -.3 0 add \rx mul -.5 -8 add)
 (! 0.3 0 add \rx mul 0.5 -8 add)
\psframe[framearc=0.4](! -.3 1 add \rx mul -.5 -8 add)
 (! 0.3 3 add \rx mul 0.5 -8 add)
\psframe[framearc=0.4](! -.3 4 add \rx mul -.5 -8 add)
 (! 0.3 4 add \rx mul 0.5 -8 add)
\psframe[framearc=0.4](! -.3 5 add \rx mul -.5 -8 add)
 (! 0.3 5 add \rx mul 0.5 -8 add)
\psframe[framearc=0.4](! -.3 6 add \rx mul -.5 -8 add)
 (! 0.3 6 add \rx mul 0.5 -8 add)
\psframe[framearc=0.4](! -.3 7 add \rx mul -.5 -8 add)
 (! 0.3 11 add \rx mul 0.5 -8 add)
  \psframe[framearc=0.4](! -.3 12 add \rx mul -.5 -8 add)
  (! 0.3 12 add \rx mul 0.5 -8 add)
  \psframe[framearc=0.4](! -.3 13 add \rx mul -.5 -8 add)
  (! 0.3 13  add \rx mul 0.5 -8 add)

  % level 5

\psframe[framearc=0.4](! -.3 1 add \rx mul -.5 -10 add)
(! 0.3 2 add \rx mul 0.5 -10 add)
\psframe[framearc=0.4](! -.3 3 add \rx mul -.5 -10 add)
(! 0.3 3 add \rx mul 0.5 -10 add)
\psframe[framearc=0.4](! -.3 7 add \rx mul -.5 -10 add)
(! 0.3 8 add \rx mul 0.5 -10 add)
\psframe[framearc=0.4](! -.3 9 add \rx mul -.5 -10 add)
(! 0.3 11 add \rx mul 0.5 -10 add)

% level 6
\psframe[framearc=0.4](! -.3 1 add \rx mul -.5 -12 add)
(! 0.3 1 add \rx mul 0.5 -12 add)
\psframe[framearc=0.4](! -.3 2 add \rx mul -.5 -12 add)
(! 0.3 2 add \rx mul 0.5 -12 add)

\psframe[framearc=0.4](! -.3 7 add \rx mul -.5 -12 add)
(! 0.3 7 add \rx mul 0.5 -12 add)
\psframe[framearc=0.4](! -.3 8 add \rx mul -.5 -12 add)
(! 0.3 8 add \rx mul 0.5 -12 add)
\psframe[framearc=0.4](! -.3 9 add \rx mul -.5 -12 add)
(! 0.3 10 add \rx mul 0.5 -12 add)
\psframe[framearc=0.4](! -.3 11 add \rx mul -.5 -12 add)
(! 0.3 11 add \rx mul 0.5 -12 add)

% level 7
\psframe[framearc=0.4](! -.3 9 add \rx mul -.5 -14 add)
(! 0.3 9 add \rx mul 0.5 -14 add)
\psframe[framearc=0.4](! -.3 10 add \rx mul -.5 -14 add)
(! 0.3 10 add \rx mul 0.5 -14 add)

\psline[linewidth=2pt](! 10 \rx mul -14 0.5 add)(! 10 \rx mul -14 1.5 add)

% Linking lines
% Level position
\psline[linewidth=2pt](! 9 \rx mul -14 0.5 add)(! 9 \rx mul -14 1.5 add)

\psline[linewidth=2pt](! 1 \rx mul -12 0.5 add)(! 1 \rx mul -12 1.5 add)
\psline[linewidth=2pt](! 2 \rx mul -12 0.5 add)(! 2 \rx mul -12 1.5 add)
\psline[linewidth=2pt](! 9 \rx mul -12 0.5 add)(! 9 \rx mul -12 1.5 add)
\psline[linewidth=2pt](! 7 \rx mul -12 0.5 add)(! 7 \rx mul -12 1.5 add)
\psline[linewidth=2pt](! 8 \rx mul -12 0.5 add)(! 8 \rx mul -12 1.5 add)
\psline[linewidth=2pt](! 11 \rx mul -12 0.5 add)(! 11 \rx mul -12 1.5 add)

\psline[linewidth=2pt](! 1 \rx mul -10 0.5 add)(! 1 \rx mul -10 1.5 add)
\psline[linewidth=2pt](! 3 \rx mul -10 0.5 add)(! 3 \rx mul -10 1.5 add)
\psline[linewidth=2pt](! 7 \rx mul -10 0.5 add)(! 7 \rx mul -10 1.5 add)
\psline[linewidth=2pt](! 11 \rx mul -10 0.5 add)(! 11 \rx mul -10 1.5 add)

\psline[linewidth=2pt](! 0 \rx mul -8 0.5 add)(! 0 \rx mul -8 1.5 add)
\psline[linewidth=2pt](! 3 \rx mul -8 0.5 add)(! 3 \rx mul -8 1.5 add)
\psline[linewidth=2pt](! 4 \rx mul -8 0.5 add)(! 4 \rx mul -8 1.5 add)
\psline[linewidth=2pt](! 5 \rx mul -8 0.5 add)(! 5 \rx mul -8 1.5 add)
\psline[linewidth=2pt](! 6 \rx mul -8 0.5 add)(! 6 \rx mul -8 1.5 add)
\psline[linewidth=2pt](! 11 \rx mul -8 0.5 add)(! 11 \rx mul -8 1.5 add)
\psline[linewidth=2pt](! 12 \rx mul -8 0.5 add)(! 12 \rx mul -8 1.5 add)
\psline[linewidth=2pt](! 13 \rx mul -8 0.5 add)(! 13 \rx mul -8 1.5 add)

\psline[linewidth=2pt](! 0 \rx mul -6 0.5 add)(! 0 \rx mul -6 1.5 add)
\psline[linewidth=2pt](! 5 \rx mul -6 0.5 add)(! 5 \rx mul -6 1.5 add)
\psline[linewidth=2pt](! 6 \rx mul -6 0.5 add)(! 6 \rx mul -6 1.5 add)
\psline[linewidth=2pt](! 13 \rx mul -6 0.5 add)(! 13 \rx mul -6 1.5 add)
\psline[linewidth=2pt](! 15 \rx mul -6 0.5 add)(! 15 \rx mul -6 1.5 add)
\psline[linewidth=2pt](! 16 \rx mul -6 0.5 add)(! 16 \rx mul -6 1.5 add)

\psline[linewidth=2pt](! 0 \rx mul -4 0.5 add)(! 0 \rx mul -4 1.5 add)
\psline[linewidth=2pt](! 13 \rx mul -4 0.5 add)(! 13 \rx mul -4 1.5 add)
\psline[linewidth=2pt](! 14 \rx mul -4 0.5 add)(! 14 \rx mul -4 1.5 add)
\psline[linewidth=2pt](! 16 \rx mul -4 0.5 add)(! 16 \rx mul -4 1.5 add)

\psline[linewidth=2pt](! 0 \rx mul -2 0.5 add)(! 0 \rx mul -2 1.5 add)
\psline[linewidth=2pt](! 16 \rx mul -2 0.5 add)(! 16 \rx mul -2 1.5 add)

  % The letters

\rput(! 0 \rx mul 0){\#}
\rput(! 1 \rx mul 0){a}
\rput(! 2 \rx mul 0){b}
\rput(! 3 \rx mul 0){b}
\rput(! 4 \rx mul 0){a}
\rput(! 5 \rx mul 0){b}
\rput(! 6 \rx mul 0){a}
\rput(! 7 \rx mul 0){a}
\rput(! 8 \rx mul 0){b}
\rput(! 9 \rx mul 0){a}
\rput(! 10 \rx mul 0){b}
\rput(! 11 \rx mul 0){b}
\rput(! 12 \rx mul 0){a}
\rput(! 13 \rx mul 0){b}
\rput(! 14 \rx mul 0){a}
\rput(! 15 \rx mul 0){a}
\rput(! 16 \rx mul 0){b}

\rput(! 0 \rx mul -2){\#}
\rput(! 1 \rx mul -2){a}
\rput(! 2 \rx mul -2){b}
\rput(! 3 \rx mul -2){b}
\rput(! 4 \rx mul -2){a}
\rput(! 5 \rx mul -2){b}
\rput(! 6 \rx mul -2){a}
\rput(! 7 \rx mul -2){a}
\rput(! 8 \rx mul -2){b}
\rput(! 9 \rx mul -2){a}
\rput(! 10 \rx mul -2){b}
\rput(! 11 \rx mul -2){b}
\rput(! 12 \rx mul -2){a}
\rput(! 13 \rx mul -2){b}
\rput(! 14 \rx mul -2){a}
\rput(! 15 \rx mul -2){a}
\rput(! 16 \rx mul -2){b}

\rput(! 0 \rx mul -4){\#}
\rput(! 1 \rx mul -4){a}
\rput(! 2 \rx mul -4){b}
\rput(! 3 \rx mul -4){b}
\rput(! 4 \rx mul -4){a}
\rput(! 5 \rx mul -4){b}
\rput(! 6 \rx mul -4){a}
\rput(! 7 \rx mul -4){a}
\rput(! 8 \rx mul -4){b}
\rput(! 9 \rx mul -4){a}
\rput(! 10 \rx mul -4){b}
\rput(! 11 \rx mul -4){b}
\rput(! 12 \rx mul -4){a}
\rput(! 13 \rx mul -4){b}
\rput(! 14 \rx mul -4){a}
\rput(! 15 \rx mul -4){a}
\rput(! 16 \rx mul -4){b}

\rput(! 0 \rx mul -6){\#}
\rput(! 1 \rx mul -6){a}
\rput(! 2 \rx mul -6){b}
\rput(! 3 \rx mul -6){b}
\rput(! 4 \rx mul -6){a}
\rput(! 5 \rx mul -6){b}
\rput(! 6 \rx mul -6){a}
\rput(! 7 \rx mul -6){a}
\rput(! 8 \rx mul -6){b}
\rput(! 9 \rx mul -6){a}
\rput(! 10 \rx mul -6){b}
\rput(! 11 \rx mul -6){b}
\rput(! 12 \rx mul -6){a}
\rput(! 13 \rx mul -6){b}
%\rput(! 14 \rx mul -6){a}
\rput(! 15 \rx mul -6){a}
\rput(! 16 \rx mul -6){b}

\rput(! 0 \rx mul -8){\#}
\rput(! 1 \rx mul -8){a}
\rput(! 2 \rx mul -8){b}
\rput(! 3 \rx mul -8){b}
\rput(! 4 \rx mul -8){a}
\rput(! 5 \rx mul -8){b}
\rput(! 6 \rx mul -8){a}
\rput(! 7 \rx mul -8){a}
\rput(! 8 \rx mul -8){b}
\rput(! 9 \rx mul -8){a}
\rput(! 10 \rx mul -8){b}
\rput(! 11 \rx mul -8){b}
\rput(! 12 \rx mul -8){a}
\rput(! 13 \rx mul -8){b}
%\rput(! 14 \rx mul -8){a}
%\rput(! 15 \rx mul -8){a}
%\rput(! 16 \rx mul -8){b}

%\rput(! 0 \rx mul -10){\#}
\rput(! 1 \rx mul -10){a}
\rput(! 2 \rx mul -10){b}
\rput(! 3 \rx mul -10){b}
%\rput(! 4 \rx mul -10){a}
%\rput(! 5 \rx mul -10){b}
%\rput(! 6 \rx mul -10){a}
\rput(! 7 \rx mul -10){a}
\rput(! 8 \rx mul -10){b}
\rput(! 9 \rx mul -10){a}
\rput(! 10 \rx mul -10){b}
\rput(! 11 \rx mul -10){b}
%\rput(! 12 \rx mul -10){a}
%\rput(! 13 \rx mul -10){b}
%\rput(! 14 \rx mul -10){a}
%\rput(! 15 \rx mul -10){a}
%\rput(! 16 \rx mul -10){b}

%\rput(! 0 \rx mul -12){\#}
\rput(! 1 \rx mul -12){a}
\rput(! 2 \rx mul -12){b}
%\rput(! 3 \rx mul -12){b}
%\rput(! 4 \rx mul -12){a}
%\rput(! 5 \rx mul -12){b}
%\rput(! 6 \rx mul -12){a}
\rput(! 7 \rx mul -12){a}
\rput(! 8 \rx mul -12){b}
\rput(! 9 \rx mul -12){a}
\rput(! 10 \rx mul -12){b}
\rput(! 11 \rx mul -12){b}
%\rput(! 12 \rx mul -12){a}
%\rput(! 13 \rx mul -12){b}
%\rput(! 14 \rx mul -12){a}
%\rput(! 15 \rx mul -12){a}
%\rput(! 16 \rx mul -12){b}

%\rput(! 0 \rx mul -14){\#}
%\rput(! 1 \rx mul -14){a}
%\rput(! 2 \rx mul -14){b}
%\rput(! 3 \rx mul -14){b}
%\rput(! 4 \rx mul -14){a}
%\rput(! 5 \rx mul -14){b}
%\rput(! 6 \rx mul -14){a}
%\rput(! 7 \rx mul -14){a}
%\rput(! 8 \rx mul -14){b}
\rput(! 9 \rx mul -14){a}
\rput(! 10 \rx mul -14){b}
%\rput(! 11 \rx mul -14){b}
%\rput(! 12 \rx mul -14){a}
%\rput(! 13 \rx mul -14){b}
%\rput(! 14 \rx mul -14){a}
%\rput(! 15 \rx mul -14){a}
%\rput(! 16 \rx mul -14){b}

\rput(! -0.5 \rx mul -15.5){$\Lyn[i]$}
\rput(! 1 \rx mul -15.5){3}
\rput(! 2 \rx mul -15.5){1}
\rput(! 3 \rx mul -15.5){1}
\rput(! 4 \rx mul -15.5){2}
\rput(! 5 \rx mul -15.5){1}
\rput(! 6 \rx mul -15.5){8}
\rput(! 7 \rx mul -15.5){5}
\rput(! 8 \rx mul -15.5){1}
\rput(! 9 \rx mul -15.5){3}
\rput(! 10 \rx mul -15.5){1}
\rput(! 11 \rx mul -15.5){1}
\rput(! 12 \rx mul -15.5){2}
\rput(! 13 \rx mul -15.5){1}
\rput(! 14 \rx mul -15.5){3}
\rput(! 15 \rx mul -15.5){2}
\rput(! 16 \rx mul -15.5){1}

% Dotted lines

\psline[linestyle=dotted](! 1 \rx mul -12 -0.5 add)(! 1 \rx mul -15)
\psline[linestyle=dotted](! 2 \rx mul -12  -0.5 add)(! 2 \rx mul -15)
\psline[linestyle=dotted](! 3 \rx mul -10  -0.5 add)(! 3 \rx mul -15)
\psline[linestyle=dotted](! 4 \rx mul -8  -0.5 add)(! 4 \rx mul -15)
\psline[linestyle=dotted](! 5 \rx mul -8  -0.5 add)(! 5 \rx mul -15)
\psline[linestyle=dotted](! 6 \rx mul -8  -0.5 add)(! 6 \rx mul -15)
\psline[linestyle=dotted](! 7 \rx mul -12  -0.5 add)(! 7 \rx mul -15)
\psline[linestyle=dotted](! 8 \rx mul -12  -0.5 add)(! 8 \rx mul -15)
\psline[linestyle=dotted](! 9 \rx mul -14  -0.5 add)(! 9 \rx mul -15)
\psline[linestyle=dotted](! 10 \rx mul -14  -0.5 add)(! 10 \rx mul -15)
\psline[linestyle=dotted](! 11 \rx mul -12  -0.5 add)(! 11 \rx mul -15)
\psline[linestyle=dotted](! 12 \rx mul -8  -0.5 add)(! 12 \rx mul -15)
\psline[linestyle=dotted](! 13 \rx mul -8  -0.5 add)(! 13 \rx mul -15)
\psline[linestyle=dotted](! 14 \rx mul -4  -0.5 add)(! 14 \rx mul -15)
\psline[linestyle=dotted](! 15 \rx mul -6  -0.5 add)(! 15 \rx mul -15)
\psline[linestyle=dotted](! 16 \rx mul -6  -0.5 add)(! 16 \rx mul -15)

\rput(! 0 1 add 0.2 add \rx mul -15){\scriptsize 0}
\rput(! 1 1 add 0.2 add \rx mul -15){\scriptsize 1}
\rput(! 2 1 add 0.2 add \rx mul -15){\scriptsize 2}
\rput(! 3 1 add 0.2 add \rx mul -15){\scriptsize 3}
\rput(! 4 1 add 0.2 add \rx mul -15){\scriptsize 4}
\rput(! 5 1 add 0.2 add \rx mul -15){\scriptsize 5}
\rput(! 6 1 add 0.2 add \rx mul -15){\scriptsize 6}
\rput(! 7 1 add 0.2 add \rx mul -15){\scriptsize 7}
\rput(! 8 1 add 0.2 add \rx mul -15){\scriptsize 8}
\rput(! 9 1 add 0.2 add \rx mul -15){\scriptsize 9}
\rput(! 10 1 add 0.2 add \rx mul -15){\scriptsize 10}
\rput(! 11 1 add 0.2 add \rx mul -15){\scriptsize 11}
\rput(! 12 1 add 0.2 add \rx mul -15){\scriptsize 12}
\rput(! 13 1 add 0.2 add \rx mul -15){\scriptsize 13}
\rput(! 14 1 add 0.2 add \rx mul -15){\scriptsize 14}
\rput(! 15 1 add 0.2 add \rx mul -15){\scriptsize 15}

\end{pspicture}
\end{center}

%---------%---------%---------%---------%---------%---------%---------%--------%
\paragraph{Lyndon table}

The Lyndon table $\Lyn$ of a (non-empty) word $y$ is defined as follows. For
a position $i$ on $y$, $i=0,\dots,|y|-1$, $\Lyn[i]$ is the length of the
longest Lyndon factor of $y$ starting at $i$:
\[\Lyn[i] = \max\{\ell \mid y[i\dd i+\ell-1] \mbox{ is a Lyndon word}\}.\]

\noindent
\fbox{\begin{picture}(290,33)(-50,10)
% abb ab ab aababbab a
\put(-25,36){\makebox(20,10)[br]{\small $i$}}
\put( 0,36){\makebox(15,10)[b]{\scriptsize 0}}
\put(15,36){\makebox(15,10)[b]{\scriptsize 1}}
\put(30,36){\makebox(15,10)[b]{\scriptsize 2}}
\put(45,36){\makebox(15,10)[b]{\scriptsize 3}}
\put(60,36){\makebox(15,10)[b]{\scriptsize 4}}
\put(75,36){\makebox(15,10)[b]{\scriptsize 5}}
\put(90,36){\makebox(15,10)[b]{\scriptsize 6}}
\put(105,36){\makebox(15,10)[b]{\scriptsize 7}}
\put(120,36){\makebox(15,10)[b]{\scriptsize 8}}
\put(135,36){\makebox(15,10)[b]{\scriptsize 9}}
\put(150,36){\makebox(15,10)[b]{\scriptsize 10}}
\put(165,36){\makebox(15,10)[b]{\scriptsize 11}}
\put(180,36){\makebox(15,10)[b]{\scriptsize 12}}
\put(195,36){\makebox(15,10)[b]{\scriptsize 13}}
\put(210,36){\makebox(15,10)[b]{\scriptsize 14}}
\put(225,36){\makebox(15,10)[b]{\scriptsize 15}}
%\put(240,36){\makebox(15,10)[b]{\scriptsize 16}}

\put(-25,22){\makebox(20,10)[br]{\small $y[i]$}}
\put( 0,24){\makebox(15,10)[b]{\sa a}}
\put(15,24){\makebox(15,10)[b]{\sa b}}
\put(30,24){\makebox(15,10)[b]{\sa b}}
\put(45,24){\makebox(15,10)[b]{\sa a}}
\put(60,24){\makebox(15,10)[b]{\sa b}}
\put(75,24){\makebox(15,10)[b]{\sa a}}
\put(90,24){\makebox(15,10)[b]{\sa a}}
\put(105,24){\makebox(15,10)[b]{\sa b}}
\put(120,24){\makebox(15,10)[b]{\sa a}}
\put(135,24){\makebox(15,10)[b]{\sa b}}
\put(150,24){\makebox(15,10)[b]{\sa b}}
\put(165,24){\makebox(15,10)[b]{\sa a}}
\put(180,24){\makebox(15,10)[b]{\sa b}}
\put(195,24){\makebox(15,10)[b]{\sa a}}
\put(210,24){\makebox(15,10)[b]{\sa a}}
\put(225,24){\makebox(15,10)[b]{\sa b}}

\put(-25,10){\makebox(20,10)[br]{\small $\Lyn[i]$}}
\put( 0,12){\makebox(15,10)[b]{3}}
\put(15,12){\makebox(15,10)[b]{1}}
\put(30,12){\makebox(15,10)[b]{1}}
\put(45,12){\makebox(15,10)[b]{2}}
\put(60,12){\makebox(15,10)[b]{1}}
\put(75,12){\makebox(15,10)[b]{8}}
\put(90,12){\makebox(15,10)[b]{5}}
\put(105,12){\makebox(15,10)[b]{1}}
\put(120,12){\makebox(15,10)[b]{3}}
\put(135,12){\makebox(15,10)[b]{1}}
\put(150,12){\makebox(15,10)[b]{1}}
\put(165,12){\makebox(15,10)[b]{2}}
\put(180,12){\makebox(15,10)[b]{1}}
\put(195,12){\makebox(15,10)[b]{3}}
\put(210,12){\makebox(15,10)[b]{2}}
\put(225,12){\makebox(15,10)[b]{1}}
\end{picture}}

\medskip
The computation of the Lyndon table is an offspring of the previous
algorithm, like the computation of the Next nearest smaller table for
Algorithm \textsc{CartesianTree}. Algorithm \textsc{LongestLyndon} computes
$\Lyn$ using the same right-to-left detection of Lyndon factors as above.

\medskip
\begin{algo}{LongestLyndon}{y \textrm{ non-empty word of length } n}
    \DOFORD{i}{n-1}{0}
        \SET{(\Lyn[i],j)}{(1,i+1)}
        \DOWHILE{j<n \mbox{ and } y[i\dd j-1] < y[j\dd j+\Lyn[j]-1]}
            \label{algo-LL:3}
            \SET{(\Lyn[i],j)}{(\Lyn[i]+\Lyn[j],j+\Lyn[j])}
        \OD
    \OD
    \RETURN{\Lyn}
\end{algo}

\section{Key property}\label{sect:kp}

It is clear that the previous algorithms all share the same algorithmic
structure. The link between the trees or their reduced versions is even
tighter when the running time of the Lyndon tree construction is concerned.
Indeed, the comparison between two consecutive phrases of the factorisation
of $y$ at line~\ref{algo-LT:4} in \textsc{LyndonTree} or at
line~\ref{algo-LL:3} in \textsc{LongestLyndon} comes back to considering the
ranks of suffixes in alphabetic order. This is shown by the next proposition
where the local comparison between two phrases is shown to be equivalent to
the comparison of their associated suffixes.

In addition, the next statement also leads to prove that the Lyndon tree of
$y$, possibly reduced to its internal nodes, has the same structure than the
Cartesion tree built from the ranks of the word suffixes, which has been
first noticed by Hohlweg and Reutenauer in \cite{HohlwegR03}.

%\paragraph{Running time of \textsc{LongestLyndon}}
%Comparison of words at line~\ref{exer142-alg-3} of algorithms can be realised
%using the ranks of suffixes based on the next lemma.

\begin{proposition}\label{prop-main}
Let $u$ be a Lyndon word and $v\cdot v_1\cdot v_2\cdots v_m$ be the Lyndon
factorisation of a word $w$. Then $u < v$ iff $uw < w$.
\end{proposition}

\begin{proof}
Let us consider the different cases.

Assume first $u < v$. If $u\sle v$ then $uw \sle vv_1v_2\cdots v_m= w$.

Consider the case where $u$ is a proper prefix of $v$. Let $e>0$ be the
largest integer for which $v=u^ez$. Since $v$ is a Lyndon word, $z$ is not
empty and we have $u^e<z$. Since $u$ is not a prefix of $z$ (by definition of
$e$) nor $z$ a prefix of $u$ (because $v$ is border-free) we have $u\sle z$.
This implies $u^{e+1}\sle u^ez = v$ and then $uw<w$.

Then assume $v\leq u$. If $v\sle u$ we have obviously $w<uw$.

It remains to consider the situations where $v$ is a prefix of $u$. If it is
a proper prefix, $u$ writes $vz$ for a non-empty word $z$. We have $v<z$
because $u$ is a Lyndon word. The word $z$ cannot be a prefix of $t=
v_1v_2\cdots v_m$ because $v$ would not be the longest Lyndon prefix of $w$,
a contradiction with a property of the factorisation. Thus, either $t\leq z$
or $z\sle t$. In the first case, if $t$ is a prefix of $z$, $w=vt$ is a
prefix of $u$ and then of $uw$, that is, $w<uw$. In the second case, for some
suffix $z'$ of $z$ and some factor $v_k$ of $t$ we have $z'\sle v_k$. The
factorisation implies $v_k \leq v$. Therefore, the suffix $z'$ of $u$ is
smaller than its prefix $v$, a contradiction with the fact that $u$ is a
Lyndon word.
\end{proof}

For each position $i$ on $y$, $i=0,\dots,|y|-1$, let $\Rank[i]$ be the rank
of the suffix $x[i\dd |y|-1]$ is the increasing alphabetic list of all
non-empty suffixes of $y$ (ranks run from $0$ to $|y|-1$).

\medskip
\noindent
\fbox{\begin{picture}(290,42)(-50,0)
% abb ab ab aababbab a
\put(-25,36){\makebox(20,10)[br]{\small $i$}}
\put( 0,36){\makebox(15,10)[b]{\scriptsize 0}}
\put(15,36){\makebox(15,10)[b]{\scriptsize 1}}
\put(30,36){\makebox(15,10)[b]{\scriptsize 2}}
\put(45,36){\makebox(15,10)[b]{\scriptsize 3}}
\put(60,36){\makebox(15,10)[b]{\scriptsize 4}}
\put(75,36){\makebox(15,10)[b]{\scriptsize 5}}
\put(90,36){\makebox(15,10)[b]{\scriptsize 6}}
\put(105,36){\makebox(15,10)[b]{\scriptsize 7}}
\put(120,36){\makebox(15,10)[b]{\scriptsize 8}}
\put(135,36){\makebox(15,10)[b]{\scriptsize 9}}
\put(150,36){\makebox(15,10)[b]{\scriptsize 10}}
\put(165,36){\makebox(15,10)[b]{\scriptsize 11}}
\put(180,36){\makebox(15,10)[b]{\scriptsize 12}}
\put(195,36){\makebox(15,10)[b]{\scriptsize 13}}
\put(210,36){\makebox(15,10)[b]{\scriptsize 14}}
\put(225,36){\makebox(15,10)[b]{\scriptsize 15}}
%\put(240,36){\makebox(15,10)[b]{\scriptsize 16}}

\put(-25,22){\makebox(20,10)[br]{\small $y[i]$}}
\put( 0,24){\makebox(15,10)[b]{\sa a}}
\put(15,24){\makebox(15,10)[b]{\sa b}}
\put(30,24){\makebox(15,10)[b]{\sa b}}
\put(45,24){\makebox(15,10)[b]{\sa a}}
\put(60,24){\makebox(15,10)[b]{\sa b}}
\put(75,24){\makebox(15,10)[b]{\sa a}}
\put(90,24){\makebox(15,10)[b]{\sa a}}
\put(105,24){\makebox(15,10)[b]{\sa b}}
\put(120,24){\makebox(15,10)[b]{\sa a}}
\put(135,24){\makebox(15,10)[b]{\sa b}}
\put(150,24){\makebox(15,10)[b]{\sa b}}
\put(165,24){\makebox(15,10)[b]{\sa a}}
\put(180,24){\makebox(15,10)[b]{\sa b}}
\put(195,24){\makebox(15,10)[b]{\sa a}}
\put(210,24){\makebox(15,10)[b]{\sa a}}
\put(225,24){\makebox(15,10)[b]{\sa b}}

\put(-25,10){\makebox(20,10)[br]{\small $\Lyn[i]$}}
\put( 0,12){\makebox(15,10)[b]{3}}
\put(15,12){\makebox(15,10)[b]{1}}
\put(30,12){\makebox(15,10)[b]{1}}
\put(45,12){\makebox(15,10)[b]{2}}
\put(60,12){\makebox(15,10)[b]{1}}
\put(75,12){\makebox(15,10)[b]{8}}
\put(90,12){\makebox(15,10)[b]{5}}
\put(105,12){\makebox(15,10)[b]{1}}
\put(120,12){\makebox(15,10)[b]{3}}
\put(135,12){\makebox(15,10)[b]{1}}
\put(150,12){\makebox(15,10)[b]{1}}
\put(165,12){\makebox(15,10)[b]{2}}
\put(180,12){\makebox(15,10)[b]{1}}
\put(195,12){\makebox(15,10)[b]{3}}
\put(210,12){\makebox(15,10)[b]{2}}
\put(225,12){\makebox(15,10)[b]{1}}

\put(-25,-2){\makebox(20,10)[br]{\small $\Rank[i]$}}
\put( 0, 0){\makebox(15,10)[b]{7}}
\put(15, 0){\makebox(15,10)[b]{15}}
\put(30, 0){\makebox(15,10)[b]{12}}
\put(45, 0){\makebox(15,10)[b]{4}}
\put(60, 0){\makebox(15,10)[b]{10}}
\put(75, 0){\makebox(15,10)[b]{1}}
\put(90, 0){\makebox(15,10)[b]{5}}
\put(105, 0){\makebox(15,10)[b]{13}}
\put(120, 0){\makebox(15,10)[b]{6}}
\put(135, 0){\makebox(15,10)[b]{14}}
\put(150, 0){\makebox(15,10)[b]{11}}
\put(165, 0){\makebox(15,10)[b]{3}}
\put(180, 0){\makebox(15,10)[b]{9}}
\put(195, 0){\makebox(15,10)[b]{0}}
\put(210, 0){\makebox(15,10)[b]{2}}
\put(225, 0){\makebox(15,10)[b]{8}}
\end{picture}}

\medskip
Applying the above property to update line~\ref{algo-LL:3}, Algorithm
\textsc{LongestLyndon} rewrites as follows, where tables $\Lyn$ and $\Rank$
concern the input word $y$.
%the inequality $x[i\dd j-1] < x[j\dd j+\Lyn[j]-1]$ at
%line~\ref{exer142-alg-3} of algorithms rewrites $\Rank[i]<\Rank[j]$.

\medskip
\begin{algo}{LongestLyndon}{y \textrm{ non-empty word of length } n}
\DOFORD{i}{n-1}{0}
  \SET{(\Lyn[i],j)}{(1,i+1)}
  \DOWHILE{j<n \mbox{ and } \Rank[i]<\Rank[j]}
    \SET{(\Lyn[i],j)}{(\Lyn[i]+\Lyn[j],j+\Lyn[j])}
  \OD
\OD \RETURN{\Lyn}
\end{algo}

\medskip
As for the running time, when the table $\Rank$ is precomputed, the
comparison of words at line~\ref{algo-LL:3} can be realised in constant time.
And since the number of comparisons is no more than $2|x|-2$ (exactly $n-1$
negative comparisons that stop the while loop and no more than $n-1$ positive
comparisons since each reduces the number of Lyndon factors in the overall
factorisation of $y$), the total running time is linear.
% $2|y|-lyn_fact(y)$

%\medskip
Note the Lyndon factorisation of $y$ can be recovered by following the
longest decreasing sequence of ranks from the first rank. It is $(7,4,3,1,0)$
in the above example, corresponding to positions $(0,3,5,7,15)$ and to the
Lyndon factorisation $\texttt{abb} \cdot \texttt{ab} \cdot \texttt{ab} \cdot
\texttt{aababbab} \cdot \texttt{a}$.

Also note the relation between $\Lyn$ and $\NNS$: $\NNS[i]=i+\Lyn[i]$, since
$\Lyn[i]$ is the smallest distance to a next rank value smaller than
$\Rank[i]$.

%---------%---------%---------%---------%---------%---------%---------%--------%
\section{Computing runs}\label{sect:runs}

Algorithm \textsc{LongestLyndon} extends to an algorithm for computing
efficiently all runs occurring in a word.

Recall that a run in the word $y$ is an occurrence of a factor, say $y[i\dd
j]$, whose length is at least twice its (smallest) period. The main result in
\cite{BannaiIINTT15} shows that a run can be identified with a special
position $s$ on $y$ for which $\Lyn[s]$ is the period of $y[i\dd j]$ and
$2\times\Lyn[s] \leq j-i+1$, considering some alphabet ordering or its
inverse.

To compute all runs of the word $y$, we just have to check if the longest
Lyndon factor starting at $i$ produces a special position of a run. This is
done by extending the Lyndon factor to the left and to the right according to
the period of the resulting factor and using Longest common extensions. This is
done by computing $r=\LCE_{\rm R}(i,i+\Lyn[i])$ and $\ell=\LCE_{\rm
L}(i-1,i+\Lyn[i]-1)$ when appropriate and verifying if $\ell+r\geq \Lyn[i]$.
If the inequality holds a run can be reported. In the algorithm below we
assume $\ell$ to be set to null if $i=0$ and $r$ to null also if
$i+\Lyn[i]=n$.

\medskip
\begin{algo}{Runs}{y \textrm{ non-empty word of length } n}
  \DOFORD{i}{n-1}{0}
    \SET{(\Lyn[i],j)}{(1,i+1)}
    \DOWHILE{j<n \mbox{ and } \Rank[i] < \Rank[j]}\label{alg2-3}
        \label{exer142-alg-3}
      \SET{(\Lyn[i],j)}{(\Lyn[i]+\Lyn[j],j+\Lyn[j])}
    \OD
    \SET{(\ell,r)}{(\LCE_{\rm L}(i-1,i+\Lyn[i]-1),\LCE_{\rm R}(i,i+\Lyn[i]))}
        \label{algo-Runs:5}
    \IF{\ell+r \geq \Lyn[i]}
%    \OCC{\ell+r \geq \Lyn[i]}
      \ACT"{output run $x[i-\ell\dd i+\Lyn[i]+r-1]$}
%      \ACT"{run of period $\Lyn[i]$ from $i-\ell$ to $i+\Lyn[i]+r-1$}
    \FI
  \OD
\end{algo}

\medskip
To locate all runs, Algorithm \textsc{Runs} has to be executed twice, for the
tables corresponding to some alphabet ordering and for the tables
corresponding to the inverse alphabet ordering.

\paragraph{Running time of \textsc{Runs}}

Algorithm \textsc{Runs} can be implemented to run in linear time $O(|y|)$
when the alphabet is linearly-sortable.

Indeed, with the hypothesis, it is known that suffixes of $y$ can be sorted
in linear time (see for example \cite{CHL07cup}). Then also the table $\Rank$
that is just the inverse of the sorted list of starting positions of the
suffixes.

Again with the hypothesis, LCE queries at line~\ref{algo-Runs:5} can be
executed in constant time after a linear-time preprocessing. The reader can
refer to the review by Fischer and Heun \cite{FischerH06} concerning LCE
queries. More advanced techniques to implement them over a general alphabet
and to compute runs can be found in
\cite{GawrychowskiKRW16,CrochemoreIKKPR16} and references therein.

Therefore the whole algorithm \textsc{Runs} runs in linear time when the
alphabet is linearly-sortable.

%---------%---------%---------%---------%---------%---------%---------%--------%
\section{Concluding remarks}

The relation between suffix sorting, part of the suffix array, and Lyndon
factorisation is examined by Mantaci, Restivo, Rosone and Sciortino in
\cite{MantaciRRS14}. Franek, Islam, Rahman and Smyth present several
algorithms to compute the Lyndon table in \cite{FranekIRS16}.

The structure of the Cartesian tree with its nodes labelled by numbers is
richer than the structure of the Lyndon tree because it seems difficult to
recover the labels without completely sorting the ranks of suffixes. This
question is certainly related to the application of Cartesian to sorting (see
for example \cite{LevcopoulosP89}).

%---------%---------%---------%---------%---------%---------%---------%--------%
%\bibliographystyle{abbrv}
%\bibliography{CartLynTrees}

\end{document}